%% file: main.tex
\input{header}

\begin{document}

\date{}
\maketitle
\begin{abstract}
    We revisit the complexity analysis of the recursive version of the randomized greedy algorithm for computing a maximal independent set (MIS), originally analyzed by Yoshida, Yamamoto, and Ito (2009).
    They showed that, on average per vertex, the expected number of recursive calls made by this algorithm is upper bounded by the average degree of the input graph. 
    While their analysis is clever and intricate, we provide a significantly simpler alternative that achieves the same guarantee.

    Our analysis is inspired by the recent work of Dalirrooyfard, Makarychev, and Mitrović (2024), who developed a potential-function-based argument to analyze a new algorithm for correlation clustering.
    We adapt this approach to the MIS setting, yielding a more direct and arguably more transparent analysis of the recursive randomized greedy MIS algorithm.
\end{abstract}

\section{Introduction}
Given a graph $G = (V, E)$, a subset of vertices $I$ is called \emph{independent} if it contains no pair of adjacent vertices.
If, in addition, $I$ cannot be augmented by adding new vertices to it while remaining independent, $I$ is a \emph{maximal independent set} (MIS).
While the maximum independent set problem is NP-hard, MIS admits a simple linear-time greedy algorithm.
In the context of sublinear-time algorithms, in 2008, Nguyen and Onak~\cite{nguyen2008constant} asked how efficiently one can determine whether a given vertex belongs to a maximal independent set (MIS).\footnote{In this setting, we want to ensure the so-called consistency. Namely, if this algorithm is executed independently on each vertex of $V$ in parallel, then those vertices for which the algorithm says they are in MIS, they altogether form an MIS of the input graph.}
For that purpose, they studied the \emph{randomized greedy MIS} (RGMIS) approach -- outlined as Algorithm~\ref{alg:RGMIS-sequential} -- and analyzed its recursive version, which can be phrased as follows.
Given a vertex $v$ and a permutation $\pi$, recursively test whether any of the neighbors $u$ of $v$ with $\pi(u) < \pi(v)$ is in MIS. 
If it is, return that $v$ is not in MIS; otherwise, return that $v$ is in MIS.
Nguyen and Onak showed that, on graphs with maximum degree $\Delta$, this recursive algorithm, when executed from an arbitrary vertex, makes at most $2^{O(\Delta)}$ recursive calls in expectation, where the expectation is taken over all permutations $\pi$.

\begin{algorithm}[h]
    \begin{algorithmic}[1]
    \State  Let $\pi : V \to \{1, \ldots, n\}$ be a random permutation of vertices.

    \State  $I \gets \emptyset$

    \For{$t = 1 \ldots n$}
        \If{no neighbor of $\pi^{-1}(t)$ is in $I$}
            \State  Add $\pi^{-1}(t)$ to $I$: $I\gets I\cup\{\pi^{-1}(t)\}$.
        \EndIf
    \EndFor

    \State \Return $I$      

    \end{algorithmic}

    \caption{\phantom{\Big|}The randomized greedy MIS algorithm. Returns the MIS produced by the randomized greedy process.
    \label{alg:RGMIS-sequential}}
\end{algorithm}

In the same work, Nguyen and Onak conjectured that the following natural heuristic optimization yields a faster algorithm for testing whether a vertex is in MIS or not: in the recursive test, $v$ examines its neighbors in increasing order with respect to $\pi$ and returns that $v$ is not in MIS upon encountering \textbf{the first neighbor that has already been included in MIS}.
This version is given in Algorithm~\ref{alg:RGMIS}. 

\begin{algorithm}[h]
    \begin{algorithmic}[1]
    \Statex \textbf{function} \FRGMIS{$(v, \pi)$}
    \State Let $w_1, \ldots, w_r$ denote the neighbors $w$ of $v$ such that $\pi(w) < \pi(v)$, ordered increasingly by $\pi(w)$. 
    \For{$i = 1 \ldots r$}
        \If{$\FRGMIS(w_i, \pi) = \true$}
            \State  \Return \false
        \EndIf
    \EndFor
    \State \Return \true
    \end{algorithmic}
    \caption{\phantom{\Big|}A recursive version of a randomized greedy MIS algorithm (Algorithm~\ref{alg:RGMIS-sequential}) with a greedy heuristic optimization.
    The algorithm returns whether a given vertex $v$ is in the MIS or not.
    \label{alg:RGMIS}}
\end{algorithm}
In 2009, Yoshida, Yamamoto, and Ito~\cite{yoshida2009improved} provided an ingenious analysis showing that this conjecture is not only true, but executing Algorithm~\ref{alg:RGMIS} from \textbf{every} vertex in the graph requires at most $m$ recursive calls in total, in expectation over $\pi$.
That result is typically stated in the literature as an average-case guarantee: in expectation over all vertices $v \in V$ and all $\pi$, $\FRGMIS(v, \pi)$ performs at most $m / n \le \Delta/2$ recursive calls.

Over recent years, the analysis and result of Yoshida, Yamamoto and Ito has been used in and has inspired several important discoveries, including developing sublinear-time algorithms for approximating the minimum vertex cover~\cite{OnakRRR12} and the maximum matching size~\cite{behnezhad2022time}, designing a more efficient MPC algorithm for maximal matching~\cite{behnezhad2019exponentially}, obtaining a very fast parallel algorithm for almost $3$-approximate correlation clustering~\cite{behnezhad2022almost}, analyzing average sensitivity of graph algorithms~\cite{varma2023average}, and developing an algorithm for stochastic matching~\cite{azarmehr2025stochastic}.
For instance, Behnezhad, Charikar, Ma, and Tan~\cite{behnezhad2022almost} build, in a clever way, on the analysis of \cite{yoshida2009improved} to show that a truncated version of the celebrated Pivot algorithm~\cite{ailon2008aggregating} yields an almost $3$-approximation for correlation clustering.\footnote{RGMIS and Pivot are the same algorithm, known as RGMIS when used to compute an MIS and as Pivot when used to compute a 3-approximate correlation clustering.}
This truncation is the key reason why their algorithm can be implemented in small depth in parallel settings and in a small number of passes in the semi-streaming model.
We believe that our work opens avenues for novel results and alternative analyses for these and related problems.

\subsection{Our contribution}

The analysis of RGMIS by Yoshida, Yamamoto, and Ito~\cite{yoshida2009improved} is clever but rather intricate. In this paper, we present a simpler analysis of the algorithm that yields the matching upper bound on the number of recursive calls. Our approach builds on the recent work of Dalirrooyfard, Makarychev, and Mitrovi{\'c}~\cite{dalirrooyfard2024pruned}, who introduced a Pruned Pivot algorithm for Correlation Clustering and analyzed it using a specially crafted potential function. We adapt their analysis to RGMIS and slightly modify the potential function to obtain a stronger bound.

As the main technical result, we prove the following.
\begin{theorem}\label{thm:main-edge-query}
    Let $G = (V, E)$ be a graph and let $\pi : V \to \{1, \ldots, n\}$ be a uniformly random permutation of its vertices.
    Then, across all invocations $\FRGMIS(u, \pi)$ for $u \in V$, the following holds:
    for every \emph{ordered} pair $(a, b)$ with $(a,b) \in E$, the expected number of times $\FRGMIS(b, \pi)$ directly invokes $\FRGMIS(a, \pi)$ is at most 
$1/2$.
\end{theorem}

\begin{corollary}
    Let $G = (V, E)$ be a graph, let $\pi : V \to \{1, \ldots, n\}$ be a uniformly random permutation of its vertices, and let $u$ be a uniformly random vertex in $V$.
    Then, in expectation over the choice of $u$ and $\pi$, the number of recursive calls made by $\FRGMIS(u, \pi)$ is at most $m/n$.\footnote{This does not count the top-level invocations of \FRGMIS, but only its recursive calls. If the top-level invocations are counted as well, then the complexity becomes $m/n + 1$.}
\end{corollary}

\section{Counting Recursive Calls using Query Paths}\label{sec:query-paths}

We first show how to bound the expected number of recursive calls in Algorithm~\ref{alg:RGMIS} by analyzing Algorithm~\ref{alg:RGMIS-sequential}. Consider the recursive calls made by Algorithm~\ref{alg:RGMIS} when it is invoked on a vertex~$u$. Suppose the algorithm makes a recursive call from $\FRGMIS(v,\pi)$ to $\FRGMIS(w,\pi)$. We define the oriented edge $(w, v)$ as a \emph{query edge} (although it may seem more natural to call $(v, w)$ a query edge, we adopt this definition for convenience later on).
It is easy to see that any sequence of recursive calls with parameters $u = u_k, u_{k-1}, \dots, u_1$ corresponds to a directed path $(u_1,\dots,u_k)$ in which each edge is a query edge. We refer to such paths as \emph{query paths}; we define this notion formally in the following section. Then, the number of recursive calls for vertex~$u$ equals the number of query paths that end at~$u$. To count recursive calls, we are going to consider the sequential version of Algorithm~\ref{alg:RGMIS} presented in Algorithm~\ref{alg:RGMIS-sequential-early-break}.

In Algorithm~\ref{alg:RGMIS-sequential}, we did not specify how to determine whether any neighbor of $\pi^{-1}(t)$ is in~$I$. This detail is crucial for our proof: for the proof to go through, we need to consider the neighbors of $\pi^{-1}(t)$ in the order given by~$\pi$, as shown in Algorithm~\ref{alg:RGMIS-sequential-early-break}.
Algorithm~\ref{alg:RGMIS-sequential-early-break} implements the same greedy heuristic but differ in the way they construct the membership dynamic table $\FRGMIS(v, \pi)$: Algorithm~\ref{alg:RGMIS-sequential-early-break} uses a  bottom-up approach, while Algorithm~\ref{alg:RGMIS} uses a top-down approach. Thus, $\FRGMIS(v, \pi)$ returns \texttt{true} if and only if Algorithm~\ref{alg:RGMIS-sequential-early-break} includes $v$ in the set~$I$.
\begin{algorithm}[t]
    \begin{algorithmic}[1]
    \State Let $\pi : V \to \{1, \ldots, n\}$ be a random permutation of vertices.

    \State $I_0 \gets \emptyset$

    \For{$t = 1 \ldots n$}
        \State  $v \gets \pi^{-1}(t)$

        \State  Sort the neighbors of $v$ with respect to their $\pi$ values.

        \For{each neighbor $w$ of $v$ with $\pi(w) < \pi(v)$}
            \If{$w \in I_{t-1}$}
                \State  $I_t \gets I_{t-1}$

                \State  Break from the inner for-loop
            \EndIf
        \EndFor
        \If{no neighbor of $v$ is in $I_{t-1}$}
            \State Add $v$ to $I_t$: $I_t\gets I_{t-1}\cup\{v\}$.
        \EndIf
    \EndFor

    \State \Return $I_n$      

    \end{algorithmic}

    \caption{\phantom{\Big|}The sequential version of 
    Algorithm~\ref{alg:RGMIS}, i.e., the randomized greedy MIS algorithm with the heuristic optimization. 
    Returns the MIS produced by the randomized greedy process.
    \label{alg:RGMIS-sequential-early-break}}
\end{algorithm}

Note that we can identify query paths by observing the behavior of Algorithm~\ref{alg:RGMIS-sequential-early-break} rather than that of Algorithm~\ref{alg:RGMIS}: an edge $(w, v)$ is a query edge if Algorithm~\ref{alg:RGMIS-sequential-early-break} examines vertex~$w$ during the iteration $\pi(v)$ -- the iteration in which it determines whether $v$ is in the MIS. Thus, we can count the number of recursive calls made by Algorithm~\ref{alg:RGMIS} when invoked on vertex~$u$ by running Algorithm~\ref{alg:RGMIS-sequential-early-break}, marking the query edges, and counting the number of directed query paths that end at vertex~$u$.

\section{Bounding the Number of Query Paths}

In this section, we bound the expected number of query paths generated by Algorithm~\ref{alg:RGMIS-sequential-early-break}. We assume that the algorithm performs iteration~$t$ at time~$t$, at which point the value $\pi^{-1}(t)$ is revealed. We denote the randomness revealed in the first $t$ iterations by ${\cal F}_t = (\pi^{-1}(1), \dots, \pi^{-1}(t))$. Note that the behavior of the algorithm during the first $t$ iterations depends only on ${\cal F}_t$.

We begin by defining two notions -- \emph{query paths} and \emph{dangerous paths} at time $t$ -- that play a crucial role in our analysis. 
Loosely speaking, a path $(u_1, \ldots, u_{k-1}, u_k)$ in $G$ is a query path at time~$t$ if, at that time, we can be certain that $u_k$ has already queried or will eventually query $u_{k-1}$, and for every $i \in \{2, \ldots, k-1\}$, $u_i$ has already queried $u_{i-1}$.
A path $(u_1, \ldots, u_{k-1}, u_k, z)$ is \emph{dangerous} at time~$t$ if $(u_1, \ldots, u_{k-1}, u_k)$ is a query path at time~$t$, and $z$ will query $u_k$ with probability strictly between 0 and 1 (conditioned on the randomness ${\cal F}_t$ revealed by time~$t$). Below, we present the formal definitions.

\begin{definition}[Query path]\label{def:query-path}
A path $(u_1, \ldots, u_k)$ with $k \ge 2$ is called a \emph{query path} at time $t$ if the following conditions hold:
\begin{itemize}
    \item If $k \geq 3$, then $(u_1, \ldots, u_{k-1})$ is a query path at time $t$;
    \item $\pi(u_{k-1}) \le t$ and $\pi(u_{k-1}) < \pi(u_k)$; and
    \item No neighbor $w$ of $u_k$ with $\pi(w) < \pi(u_{k-1})$ belongs to the MIS $I_t$ maintained by the algorithm.
\end{itemize}
\end{definition}

\begin{remark}\label{rem:def:query-path}
The sets $I_t$ are monotone with respect to inclusion; that is, 
$I_t \subseteq I_{t+1}$ for all $t$. In particular, once a vertex $w$ is added 
to $I_t$, it remains in $I_{t'}$ for every $t' \ge t$. Moreover, a vertex $w$ 
can only be added to the independent set at time $t = \pi(w)$.
Therefore, the last condition in Definition~\ref{def:query-path} can be 
equivalently restated as:
\begin{itemize}
    \item No neighbor $w$ of $u_k$ with $\pi(w) < \pi(u_{k-1})$ is added to the
    MIS at time $\pi(w)$; that is, $w \notin I_{\pi(w)}$.
\end{itemize}
\end{remark}

Observe that whether a path $P$ is a query path at time $t$ is completely determined by the randomness revealed up to time $t$, namely, the set of vertices $v$ with $\pi(v) \le t$ and their ranks. In other words, the event that $P$ is a query path at time $t$ is measurable with respect to the filtration~$\mathcal{F}_t$. 
Note that the rank of $u_k$ may or may not have been revealed by time $t$. In the latter case, although the exact value of $\pi(u_k)$ is unknown at time $t$, we do know that $\pi(u_k) > t$.

We now show that once a path $P=(u_1,\ldots,u_k)$ becomes a query path at some 
time $t$, it remains a query path throughout the remainder of the algorithm. 
Moreover, if $P$ ever becomes a query path, then it does so at time 
$\pi(u_{k-1})$.

\begin{claim}\label{cl:query-path-t-prime}
If $P = (u_1,\ldots,u_k)$ is a query path at time $t$, then it is also a 
query path at every time\footnote{In this claim, we do not assume that 
$t'<t$ or $t'>t$.} $t' \ge \pi(u_{k-1})$.
\end{claim}
\begin{proof}
We prove this claim by induction on $k$. Since $P$ is a query path at time $t$, 
the prefix $(u_1,\ldots,u_{k-1})$ is a query path at time $t$. By the induction 
hypothesis, it is also a query path at time~$t'$. The condition 
$\pi(u_{k-1}) < \pi(u_k)$ holds because $P$ is a query path at time $t$, and 
the condition $\pi(u_{k-1}) \le t'$ holds by assumption. 
It remains to verify the third condition. By Remark~\ref{rem:def:query-path}, 
this condition is equivalent to requiring that no neighbor $w$ of $u_k$ with 
$\pi(w) < \pi(u_{k-1})$ satisfies $w \in I_{\pi(w)}$. Since this is a 
time-independent condition (its truth depends only on the sets $I_{\pi(w)}$, 
not on $t$ or $t'$), it holds at time $t'$ whenever it holds at time $t$.  Thus all three conditions in Definition~\ref{def:query-path} hold at time $t'$, 
and therefore $P$ is a query path at time $t'$.
\end{proof}

By the previous claim, if a path is a query path at some time $t$, then it is also a query path at time $\pi(u_k)$. At that time the algorithm considers vertex $u_k$, and since no neighbor $w$ with $\pi(w) < \pi(u_{k-1})$ belongs to the independent set $I_t$, it queries vertex $u_{k-1}$. Thus, if a path $(u_1,\dots,u_k)$ is a query path at time $t$, then the following two conditions hold:
\begin{itemize}
 \item If $k \ge 3$, then $(u_1, \ldots, u_{k-1})$ is a query path at time $t$; and
 \item Conditioned on the current state of the algorithm, vertex $u_k$ queries $u_{k-1}$ with probability $1$, that is,
 \[
 \Pr\!\left[u_k \text{ queries } u_{k-1} \mid \mathcal{F}_t\right] = 1.
 \]
\end{itemize}
The converse statement also holds: if a path satisfies the above conditions, then it is a query path. Indeed, suppose that $\pi(u_{k-1}) > t$, meaning that the rank of $u_{k-1}$ has not been revealed by time $t$. Then, conditioned on $\mathcal{F}_t$, there is a positive probability that $\pi(u_{k-1}) > \pi(u_k)$. In that case, $u_k$ would not query $u_{k-1}$. 
Similarly, if there exists a neighbor $w$ of $u_k$ with $\pi(w) < \pi(u_{k-1})$ that belongs to $I_t$, then $u_k$ would certainly not query $u_{k-1}$.

If we replace the second condition with the requirement that this probability is strictly between $0$ and $1$, we obtain the definition of a dangerous path.

\begin{definition}[Dangerous path]
\label{def:dangerous}
A path $P = (u_1, \ldots, u_k, z)$ with $k \ge 1$ is called a 
\emph{dangerous path} at time $t$ if the following conditions hold:
\begin{itemize}
    \item If $k \ge 2$, then $(u_1, \ldots, u_k)$ is a query path at time $t$; and
        \item Conditioned on the current state of the algorithm ${\cal F}_t$, the event that 
$P$ will eventually become a query path (that is, that $z$ will query $u_k$ at 
some future time $t' > t$) has probability strictly between $0$ and $1$.
    \end{itemize}
\end{definition}

As in the case of query paths, whether $P$ is a dangerous path at time $t$ is determined by the randomness revealed up to time $t$. 
Equivalently, the event that $P$ is dangerous at time $t$ is measurable with respect to $\mathcal F_t$.

Our goal is to upper bound the expected total number of query paths in $G$ at time~$n$. As we will see, at time~$0$, the graph contains no query paths but does contain $2m$ dangerous paths. We now examine how query and dangerous paths evolve over time and show that the number of the former plus half the number of the latter does not increase in expectation. As a result, the expected number of query paths after the final iteration of the algorithm is at most $m$. We begin by establishing the following simple relationship between the evolutions of query and dangerous paths.

\begin{claim}[Evolution of query and dangerous paths]\label{claim:query-dangerous-properties}
    Let $\hP = (u_1, \ldots, u_k, z, w)$ be a path and $P = (u_1, \ldots, u_k, z)$ be its prefix.
    \begin{enumerate}
        \item\label{item:query-paths-come-from-dangerous} If $P$ is a newly created query path at time $t+1$, then $P$ is dangerous at time $t$.

        \item\label{item:new-dangerous-path-had-dangerous-prefix-a-moment-ago} If $\hP$ is a newly created dangerous path at time $t+1$, then $P$ is dangerous at time $t$.
    \end{enumerate}
\end{claim}
\begin{proof}
    We prove each of the claims separately.

\medskip
    
\paragraph{Proof of \eqref{item:query-paths-come-from-dangerous}.}
Let $P = (u_1, \ldots, u_k, z)$. By the definition of a query path, at time 
$t+1$ we have $\pi(u_1) < \cdots < \pi(u_k) < \pi(z)$ and $\pi(u_k) \le t+1$. 
In particular, $\pi(u_{k-1}) < \pi(u_k) \le t+1$, and hence $\pi(u_{k-1}) \le t$.  
By Claim~\ref{cl:query-path-t-prime}, the prefix $(u_1,\ldots,u_k)$ is therefore 
a query path at time $t$.

To show that $P$ was a dangerous path at time $t$, it remains to prove that the 
probability that $P$ eventually becomes a query path was strictly less than $1$ 
at time $t$. We first observe that $\pi(u_k) = t+1$. Indeed, if $\pi(u_k) \le t$, 
then $P$ would already have been a query path at time $t$, contradicting the 
choice of $t$. Thus, the value $\pi(u_k)$ had not yet been revealed at time $t$. 
Similarly, the value $\pi(z)$ had not been revealed at time $t$. Therefore, at 
time $t$ there was a positive probability that $\pi(z)=t+1$ while 
$\pi(u_k)>t+1$. In this case we would have $\pi(z)<\pi(u_k)$, and hence $z$ 
would never query $u_k$, implying that $P$ would never become a query path.

Thus, conditioned on $\mathcal{F}_t$, the event that $P$ eventually becomes a 
query path has probability strictly less than $1$, and so $P$ was a dangerous 
path at time $t$.

\paragraph{Proof of \eqref{item:new-dangerous-path-had-dangerous-prefix-a-moment-ago}.}
Suppose that $\hP = (u_1,\ldots,u_k,z,w)$ is a dangerous path at time $t+1$. 
By Definition~\ref{def:dangerous}, the prefix $(u_1,\ldots,u_k,z)$ is a query path at time $t+1$, and conditioned on $\mathcal F_{t+1}$ the probability that $\hP$ eventually becomes a query path lies strictly between $0$ and $1$. Observe that the conditional probability that $\hP$ eventually becomes a query path given $\mathcal F_t$ was also strictly between $0$ and $1$  at time $t$. Since $\hP$ is newly created at time $t+1$, it was not dangerous at time $t$. 
Therefore the prefix $(u_1,\ldots,u_k,z)$ could not have been a query path at time $t$. 
Consequently, $(u_1,\ldots,u_k,z)$ is a newly created query path at time $t+1$. 
By Part~\eqref{item:query-paths-come-from-dangerous}, this implies that the prefix $P=(u_1,\ldots,u_k,z)$ was a dangerous path at time $t$, as required.
\end{proof}

\begin{claim}\label{claim:vertex-rank-properties}
    Let $P = (u_1, \ldots, u_k, z)$ be a dangerous path at time $t$. Then the following hold:
    \begin{enumerate}[(i)]
        \item\label{item:pi-u_k-t} $\pi(u_k) > t$ and $\pi(z) > t$, and

        \item\label{item:no-neighbor-of-w-is-in-MIS} No neighbor $v$ of $z$ with $\pi(v) \le t$ belongs to the independent set $I_t$ maintained by the algorithm.
    \end{enumerate}
\end{claim}

\begin{proof}

Assume, toward a contradiction, that $\pi(u_k) \le t$. Then, by time~$t$, we would be able to determine whether $z$ queries $u_k$ at some time $t' > t$. Specifically, if there exists a neighbor $v$ of $z$ with $\pi(v) \le \pi(u_k)$ and $v \in I_t$, then $z$ does not query $u_k$, since at time $\pi(v) \le t$ we already know that $z$ is not in the independent set. If, on the other hand, no such $v$ exists, then $z$ queries $u_k$ at time $t' = \pi(z) > \pi(u_k)$. In either case, the conditional probability (given $\mathcal{F}_t$) that $P$ eventually becomes a query path is either $0$ or $1$, contradicting the assumption that $P$ is dangerous at time $t$. Hence, $\pi(u_k) > t$.

Similarly, if there exists a neighbor $v$ of $z$ with $\pi(v) \le t$ and $v \in I_t$, then by time~$t$ we know that $z$ does not belong to the maximal independent set, and therefore $z$ will never query $u_k$ at any time $t' > t$. This again contradicts the assumption that $P$ is dangerous at time $t$, and proves part~(ii).

Finally, for $P$ to have a nonzero probability of becoming a query path, we must also have $\pi(z) > t$. Indeed, if $\pi(z) \le t$, then since $\pi(u_k) > t$, we would have $\pi(z) < \pi(u_k)$, and thus $z$ would never query $u_k$. This would again imply that $P$ is not dangerous at time $t$, completing the proof of part~(i).
\end{proof}

We now show that the following potential function -- the number of query paths plus half the number of dangerous paths -- does not increase in expectation.
To this end, for every \emph{directed} edge $(a, b)$ in the graph, we analyze the evolution of the following quantity:
\begin{multline*}  
    \Phi_t(a, b) \eqdef \#[\text{query paths beginning with $(a, b)$ at time $t$}] +\\ + 1/2 \cdot \#[\text{dangerous paths beginning with $(a, b)$ at time $t$}].
\end{multline*}
We prove that $\Phi_t(a,b)$ is a supermartingale with respect to the filtration $(\mathcal{F}_t)$; that is, 
$$\E{\Phi_{t+1}(a,b)\mid {\cal F}_t}\leq \Phi_t(a,b)$$  for all $t$.

At time $0$, every directed edge $(a,b)$ is a dangerous path, since there is a non-zero probability that $b$ will query $a$. At this point, the graph contains no query paths. Thus, $\Phi_0(a, b) = 1/2$.

\begin{observation}\label{obs:Phi_0}
    For each directed edge $(a, b)$, we have $\Phi_0(a, b) = 1/2$.
\end{observation}

At time $n$, there are no dangerous paths. Thus, we have the following claim. 

\begin{observation}\label{obs:Phi_n}
    For each directed edge $(a, b)$, we have
    \[
        \Phi_n(a, b) = \#[\text{query paths beginning with $(a, b)$ at time $n$}].
    \]
\end{observation}

We now prove the main technical result of this paper.

\begin{theorem}\label{thm:supermartingale}
    The process $\Phi_t(a, b)$ is a supermartingale.
\end{theorem}
This theorem and Observation~\ref{obs:Phi_0} imply that
$\E{\Phi_n(a, b)} \le \E{\Phi_0(a, b)} = 1/2$.
Together with Observation~\ref{obs:Phi_n}, this means that, in expectation, there is at most $1/2$ query path beginning with $(a, b)$ at time $n$, i.e., at time when the entire $\pi$ is revealed. 
Thus, the total number of query paths is at most
$2m \times \frac{1}{2} = m,$
in expectation, in the end of the algorithm. As discussed in Section~\ref{sec:query-paths}, the number of query paths starting with $(a, b)$ equals the total number of times $b$ queries $a$ over all invocations of $\FRGMIS(v, \pi)$ for $v \in V$. Therefore, Theorem~\ref{thm:main-edge-query} follows.

\medskip

\begin{proof}[Proof of Theorem~\ref{thm:supermartingale}]
    Let $\cQ_t(a, b)$ and $\cD_t(a, b)$ be the set of query and dangerous paths, respectively, beginning with $(a, b)$ at time $t$. By definition, we have $\Phi_t(a, b) = |\cQ_t(a, b)| + 1/2 \cdot |\cD_t(a, b)|$.
    Our goal is to analyze,
    \begin{align}\label{eq:potential-difference-expanded}
          \Exp\bigl[{\Phi_{t+1}(a, b)\ |\ {\cal F}_t}\bigr]  - \Phi_t(a, b)
         &= \notag
         \Exp\Bigl[|\cQ_{t+1}(a, b)| \ |\ {\cal F}_t\Bigr] - |\cQ_t(a, b)| \\&\phantom{=}+ 1/2 \cdot \rb{\Exp\Bigl[|\cD_{t+1}(a, b)|\ |\ {\cal F}_t\Bigr]  - |\cD_t(a, b)|}
    \end{align}
    aiming to show that this difference is at most $0$.

\paragraph{The analysis of $|\cQ_{t+1}(a, b)| - |\cQ_{t}(a, b)|$, given ${\cal F}_t$.}
By Claim~\ref{claim:query-dangerous-properties}~\eqref{item:query-paths-come-from-dangerous},  
every newly created query path at time $t+1$ must belong to $\cD_t(a, b)$;  
that is, $\cQ_{t+1}(a, b) \setminus \cQ_t(a, b) \subseteq \cD_t(a, b)$.
Consider a dangerous path $P = (a, b, u_1, \ldots, u_k, z)\in\cD_t(a, b)$.
By Claim~\ref{claim:vertex-rank-properties}(i), we have $\pi(u_k)>t$.
By the definition of a query path, $P$ can become a query path at time $t+1$
only if $\pi(u_k)\le t+1$.  
Since $\pi(u_k)>t$, it is possible exactly when $\pi(u_k)=t+1$.
Given ${\cal F}_t$, the vertex $\pi^{-1}(t+1)$ is chosen uniformly among the 
$n-t$ vertices whose ranks are not yet revealed, and therefore  
$\Pr[\pi(u_k)=t+1 \mid {\cal F}_t]=1/(n-t)$.
Hence,
\begin{equation}\label{eq:Q-increment}
    \Exp\big[|\cQ_{t+1}(a, b)| \mid {\cal F}_t\big] - |\cQ_t(a, b)| 
      = \frac{|\cD_t(a, b)|}{\,n - t\,}.
\end{equation}

\paragraph{The analysis of $|\cD_{t+1}(a, b)| - |\cD_{t}(a, b)|$, given ${\cal F}_t$.}
By Claim~\ref{claim:query-dangerous-properties}~\eqref{item:new-dangerous-path-had-dangerous-prefix-a-moment-ago},  
every newly created dangerous path 
$R' = (a, b, u_1, \ldots, u_k, z, w)$ at time $t+1$ has a prefix 
$R = (a, b, u_1, \ldots, u_k, z)$ that belongs to $\cD_t(a, b)$. 
For a fixed dangerous path $R \in \cD_t(a, b)$, let $\Delta_R$ be the net change in the number of dangerous paths caused by $R$ at time $t+1$: namely, the number of new dangerous extensions of $R$ created at time $t+1$, minus one if $R$ ceases to be dangerous (and zero otherwise). Then
\[
|\cD_{t+1}(a,b)| - |\cD_t(a,b)|
=
\sum_{R \in \cD_t(a,b)} \Delta_R.
\]

We now bound $\Exp[\Delta_R \mid \mathcal{F}_t]$. 
To this end, we state three \emph{necessary} conditions for an extension 
$R' = (a, b, u_1, \ldots, u_k, z, w)$ of $R$ to have a nonzero probability of belonging to 
$\cD_{t+1}(a, b) \setminus \cD_t(a, b)$.

Since $R$ is a dangerous path at time $t$, by 
Claim~\ref{claim:vertex-rank-properties}~\eqref{item:pi-u_k-t}, we have
\begin{equation}\tag{A}\label{item:pi-of-vertices-UZ}
\pi(u_k) > t 
\quad \text{and} \quad 
\pi(z) > t .
\end{equation}

If $R'$ is dangerous at time $t+1$, then by 
Claim~\ref{claim:vertex-rank-properties}~\eqref{item:pi-u_k-t} we must have 
$\pi(w) > t+1$. 
For this event to have positive probability conditioned on $\mathcal F_t$, 
it is necessary that
\begin{equation}\tag{B}
\label{item:pi-of-vertices-W}
\pi(w) > t .
\end{equation}

In addition, if $R'$ is dangerous at time $t+1$, then by 
Claim~\ref{claim:vertex-rank-properties}~\eqref{item:no-neighbor-of-w-is-in-MIS}, no neighbor $v$ of $w$ with $\pi(v) \le t+1$ belongs to the MIS $I_{t+1}$. 
For this event to have positive probability conditioned on $\mathcal F_t$, 
it is necessary that
\begin{equation}\tag{C}\label{item:condition-on-w-and-MIS}
\text{no neighbor $v$ of $w$ with $\pi(v) \le t$ belongs to the MIS $I_t$.}
\end{equation}
\input{fig-dangerous-path}

Let $A_R$ be the set of neighbors $w$ of $z$ (excluding $u_k$) satisfying 
Conditions~\eqref{item:pi-of-vertices-W} and~\eqref{item:condition-on-w-and-MIS}. See Figure~\ref{fig:dangerous-path}.  
By Conditions~\eqref{item:pi-of-vertices-UZ} and ~\eqref{item:pi-of-vertices-W}, none of the vertices in  
$A_R \cup \{z, u_k\}$ has its rank revealed by time $t$.  
Hence each of these vertices equals $v^* = \pi^{-1}(t+1)$ with probability $1/(n-t)$. We consider four cases:
\begin{itemize}
\item \textbf{Case $v^* = u_k$.}  
In this case, $R$ becomes a query path and therefore ceases to be dangerous.  
For each $w \in A_R$, the extension 
$(a, b, u_1, \ldots, u_k, z, w)$ may become dangerous\footnote{The only situation in which an extension 
$(a,b,u_1,\ldots,u_k,z,w)$ does not become dangerous is when both 
(i) $u_k$ is adjacent to $w$, and 
(ii) $u_k$ belongs to the MIS $I_{t+1}$. 
In particular, if there is no edge between $u_k$ and $w$, then the extension 
$(a,b,u_1,\ldots,u_k,z,w)$ necessarily becomes dangerous. 
In this case we have $\Delta_R = |A_R| - 1$.
Using this observation, one can verify that if the graph is triangle-free, then 
$\E{\Phi_{t+1}(a,b)\mid\mathcal F_t} = \Phi_t(a,b)$ for every $t$, 
so $\Phi_t(a,b)$ is a martingale (rather than merely a supermartingale). 
It follows that the bound in Theorem~\ref{thm:main-edge-query} is tight for triangle-free graphs, 
and this bound is strict whenever the graph contains a triangle.
}. 
Thus,
\[
\Delta_R \le |A_R| - 1.
\]
This case occurs with probability $1/(n-t)$.

\item \textbf{Case $v^* \in A_R$.}  
Then $v^*$ joins the MIS at time $t+1$.  
Since $v^*$ is a neighbor of $z$ with $\pi(v^*) < \pi(z)$,  
vertex $z$ will never query $u_k$.  
Hence $R$ ceases to be dangerous and no new dangerous extension of $R$ is created.  
Thus,
\[
\Delta_R = -1.
\]
This case occurs with probability $|A_R|/(n-t)$.

\item \textbf{Case $v^* = z$.}  
Here $\pi(z) < \pi(u_k)$, so $z$ will never query $u_k$.  
Again, $R$ ceases to be dangerous and no extension is created.  
Thus,
\[
\Delta_R = -1.
\]
This case occurs with probability $1/(n-t)$.

\item \textbf{Case $v^* \notin A_R \cup \{z,u_k\}$.}  
In this case, no new dangerous extension of $R$ is created and $R$ remains to be dangerous.  
Hence,
\[
\Delta_R = 0.
\]

\end{itemize}

Combining the four cases, we obtain
\[
\E{\Delta_R \mid \mathcal{F}_t}
\le
\frac{|A_R| - 1}{n-t}
-
\frac{|A_R|}{n-t}
-
\frac{1}{n-t}
=
-\frac{2}{n-t}.
\]
Summing over all $R \in \cD_t(a,b)$ gives
\[
\E{
|\cD_{t+1}(a,b)| - |\cD_t(a,b)|
\;\middle|\;
\mathcal{F}_t
}
=
\sum_{R \in \cD_t(a,b)}
\Exp[\Delta_R \mid \mathcal{F}_t]
\le
-\frac{2\,|\cD_t(a,b)|}{n-t}.
\]
Combining this inequality with Equations~\ref{eq:potential-difference-expanded} and~\ref{eq:Q-increment} 
yields
\[
    \E{
        \Phi_{t+1}(a,b) \mid \mathcal{F}_t
    }
    - \Phi_t(a,b)
    \le 0,
\]
as desired.
\end{proof}

\section*{Acknowledgements}
We are grateful to Ronitt Rubinfeld for encouraging us to write this article and for helpful discussions about this problem. 
We also thank Jane Lange for insightful conversations and the anonymous reviewers for their valuable feedback on our submission. 
We thank Yuichi Yoshida and Senem Işık for asking for which graphs our analysis yields a tight bound.

\bibliographystyle{siamplain}
\bibliography{ref}
\end{document}

%% file: header.tex
\documentclass[11pt]{article}
\usepackage[english]{babel}
\usepackage[utf8]{inputenc}
\usepackage[T1]{fontenc}
\usepackage[margin=1in]{geometry}
\usepackage{nicefrac}
\usepackage{amsmath,amsthm}
\usepackage{amsfonts}
\usepackage{mathabx}
\usepackage{graphicx}
\usepackage{todonotes}
\usepackage{color}
\usepackage[dvipsnames]{xcolor}
\usepackage{enumerate}
\usepackage{hyperref}
\hypersetup{
    unicode=false,          
    colorlinks=true,        
    linkcolor=Blue,          
    citecolor=purple,        
    filecolor=magenta,      
    urlcolor=cyan           
}

\usepackage{algorithm}
\usepackage[noend]{algpseudocode}
\usepackage[mathscr]{euscript}
\usepackage{xspace}

\newtheorem{theorem}{Theorem}[section]

\newtheorem{corollary}[theorem]{Corollary}

\newtheorem{definition}[theorem]{Definition}
\newtheorem{remark}[theorem]{Remark}

\newtheorem{observation}[theorem]{Observation}
\newtheorem{claim}[theorem]{Claim}

\usepackage[capitalize,nameinlink]{cleveref}
\usepackage{soul}
\crefname{theorem}{Theorem}{Theorems}
\Crefname{lemma}{Lemma}{Lemmas}
\Crefname{invariant}{Invariant}{Invariants}
\Crefname{claim}{Claim}{Claims}
\crefname{observation}{Observation}{Observations}
\Crefname{observation}{Observation}{Observations}
\Crefname{algorithm}{Algorithm}{Algorithms}
\Crefname{figure}{Figure}{Figures}

\newcommand{\eqdef}{\stackrel{\text{\tiny\rm def}}{=}}
\newcommand{\E}[1]{{\mathbb{E}}\left[#1\right]}

\newcommand{\Exp}{{\mathbb{E}}}

\newcommand{\rb}[1]{\left(#1\right)}

\newcommand{\cD}{\mathcal{D}}

\newcommand{\cQ}{\mathcal{Q}}
\newcommand{\hP}{\hat{P}}

\newcommand{\true}{\textsc{true}\xspace}
\newcommand{\false}{\textsc{false}\xspace}

\newcommand{\FRGMIS}{\texttt{RGMIS}\xspace}

\title{A Simple Average-case Analysis of \\ Recursive Randomized Greedy MIS}
\author{Mina Dalirrooyfard \and
Konstantin Makarychev\thanks{Supported by the NSF Awards CCF-1955351 and EECS-2216970.}
\and
Slobodan Mitrović\thanks{Supported by NSF CAREER Award No.~2340048.
Part of this work was conducted while S.M.~was visiting the Simons Institute for the Theory of Computing.}}
\date{}

%% file: fig-dangerous-path.tex
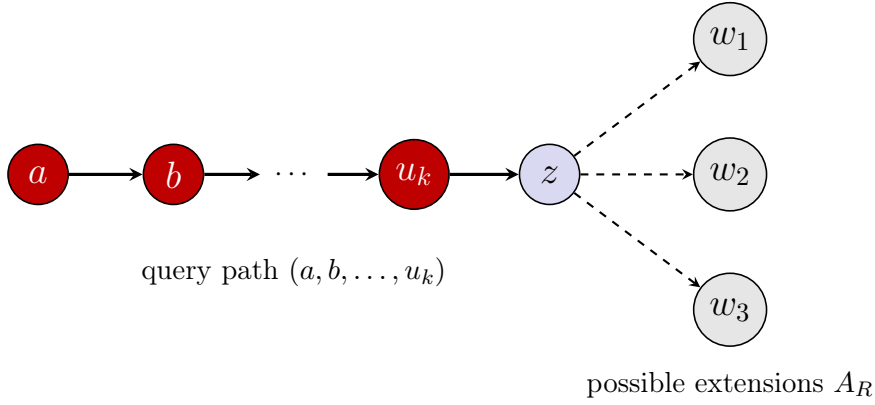
\begin{figure}[t]
\centering
\begin{tikzpicture}[
    >=stealth,
    every path/.style={->, very thick},
    every node/.style={
        circle,
        minimum size=8mm,
        draw,
        line width=0.6pt   
    },
    rednode/.style={
        fill=red!75!black,
        text=white
    },
    znode/.style={
        fill=blue!50!gray!20
    },
    wnode/.style={
        fill=gray!20
    },
]

\node[rednode] (u1) at (0,0) {\Large$a$};
\node[rednode] (u2) at (1.8,0) {\Large$b$};
\node[draw=none] (dots) at (3.4,0) {$\cdots$};
\node[rednode] (uk) at (5,0) {\Large$u_k$};
\node[znode] (z) at (6.8,0) {\Large$z$};

\draw (u1) -- (u2);
\draw (u2) -- (dots);
\draw (dots) -- (uk);
\draw (uk) -- (z);

\node[wnode] (w1) at (9.2,1.8) {\Large${w_1}$};
\node[wnode] (w2) at (9.2,0)   {\Large${w_2}$};
\node[wnode] (w3) at (9.2,-1.8){\Large${w_3}$};

\draw[dashed, thick] (z) -- (w1);
\draw[dashed, thick] (z) -- (w2);
\draw[dashed, thick] (z) -- (w3);

\node[draw=none,rectangle] at (3.4,-1.3)
    {query path $(a,b,\dots,u_k)$};

\node[draw=none,rectangle] at (9.2,-2.8)
    {possible extensions $A_R$};

\end{tikzpicture}
\caption{Illustration of a dangerous path $R=(u_1,\ldots,u_k,z)$ and its possible extensions to vertices $w \in A_R$. 
The ranks of the vertices $u_k$, $z$, and $w_i$ have not yet been revealed by time $t$. 
Hence $\pi(u_k), \pi(z), \pi(w_i) > t$, and each of these vertices receives rank $t+1$ with probability exactly $\tfrac{1}{n-t}$.
}
\label{fig:dangerous-path}
\end{figure}